\documentclass[a4paper,12pt]{article}
\pdfoutput=1
\usepackage{test,graphicx}
\usepackage{pgfplots}
\pgfplotsset{compat=1.18}

\usepackage{hyperref}
\hypersetup{colorlinks,citecolor=blue,urlcolor=magenta}
\usepackage{doi}

\usepackage[style=ext-authoryear-comp,%numeric-comp,
sorting=nyt,
dashed=false, % don't use dash for repeated author
maxcitenames=2, % maximum names to cite in text
maxbibnames=99, % maximum names to list in reference
uniquelist=false,
uniquename=false,
useprefix=true,
giveninits=true, % first name initials
natbib, % use natbib commands
date=year % year only
]{biblatex}

\AtBeginRefsection{\GenRefcontextData{sorting=ynt}}
\AtEveryCite{\localrefcontext[sorting=ynt]}

\DeclareFieldFormat{pages}{#1} % suppress pp.
\renewbibmacro{in:}{\ifentrytype{article}{}{\printtext{\bibstring{in}\intitlepunct}}} % delete In:

\DeclareFieldFormat[article,inbook,incollection,inproceedings,patent,thesis,unpublished]{titlecase:title}{\MakeSentenceCase*{#1}} % change title to sentence case

\usepackage[inline,shortlabels]{enumitem}
\setlist[enumerate,1]{label=(\roman*)}

\usepackage[margin = 1.25in]{geometry}
\usepackage{setspace}
\newcommand{\cK}{\mathcal{K}}
\newcommand{\cP}{\mathcal{P}}

\title{Comment on `Asset Bubbles and Overlapping Generations'\thanks{We thank Gadi Barlevy, Tomohiro Hirano, Gerhard Sorger, and two anonymous referees for comments and suggestions and Johar Cassa for research assistance. Toda acknowledges financial support from Japan Center for Economic Research; Japan Securities Scholarship Foundation; and the Joint Usage/Research Center, Institute of Economic Research, Hitotsubashi University (Grant ID: IERPK2515).}}
\author{Ngoc-Sang Pham\thanks{EM Normandie Business School, M\'etis Lab. Email: \href{mailto:npham@em-normandie.fr}{npham@em-normandie.fr}.} \and Alexis Akira Toda\thanks{Department of Economics, Emory University. Email: \href{mailto:alexis.akira.toda@emory.edu}{alexis.akira.toda@emory.edu}.}}
\date{\today\\
Accepted at \emph{Econometrica}}

\numberwithin{equation}{section}
\numberwithin{lem}{section}
\begin{document}
\maketitle

\begin{abstract}

\citet{Tirole1985} studied an overlapping generations model with capital accumulation and showed that the emergence of asset bubbles solves the capital over-accumulation problem. His Proposition 1(c) claims that if the dividend growth rate is above the bubbleless interest rate (the steady-state interest rate in the economy without the asset) but below the population growth rate, then bubbles are necessary in the sense that there exists no bubbleless equilibrium but there exists a unique bubbly equilibrium. We show that this result (as stated) is incorrect by presenting an example economy that satisfies all assumptions of Proposition 1(c) but its unique equilibrium is bubbleless. We also restore Proposition 1(c) under the additional assumptions that initial capital is sufficiently large and dividends are sufficiently small. We show through examples that these conditions are essential.

\medskip


\end{abstract}
\section{Introduction}

\citet{Tirole1985} studied under what conditions an asset price bubble (asset price exceeding the fundamental value defined by the present discounted value of dividends) can emerge in an overlapping generations (OLG) model with capital accumulation and showed that bubbles can solve the capital over-accumulation problem. Even after 40 years since its publication, \citet{Tirole1985}'s model remains highly relevant as it is one of the benchmark models to understand bubbles.

Although the vast majority of the subsequent literature on so-called ``rational bubbles'' has focused on bubbles attached to intrinsically worthless assets that do not pay any dividends (``pure bubble'' like fiat money or cryptocurrency),\footnote{See \citet{MartinVentura2018} and \citet{HiranoToda2024JME} for literature reviews on rational bubbles. For other approaches including heterogeneous beliefs and asymmetric information, see \citet{BrunnermeierOehmke2013} and \citet{Barlevy2025BubbleBook}.} \citet{Tirole1985}'s original analysis actually contains a discussion of bubbles attached to a dividend-paying asset. To state his result, let $G$ be the economic (population) growth rate, $G_d$ the dividend growth rate, and $R$ the steady-state interest rate in the absence of the asset.\footnote{Our notation corresponds to \citet{Tirole1985}'s by setting $G=1+n$, $G_d=1$, and $R=1+\bar{r}$.} Proposition 1(c) of \citet{Tirole1985} claims that if $R<G_d<G$, there exists no bubbleless equilibrium, there exists a unique bubbly equilibrium, which is asymptotically bubbly and the interest rate converges to $G$.

This comment provides a counterexample to Proposition 1(c) of \citet{Tirole1985} but restores it under additional assumptions. We immediately point out that the subsequent literature has overwhelmingly cited \citet{Tirole1985} not in the context of Proposition 1(c) but in the context of pure bubbles (bubbles attached to assets with zero dividends), which is valid. Thus, we do not dispute the key contribution of \citet{Tirole1985}. The reader may wonder why it is necessary to raise issues with a result that has been overlooked for four decades. The main reason is that there are limitations to pure bubble models including lack of realism, equilibrium indeterminacy, and inability to connect to the econometric literature that uses the price-dividend ratio (see, for instance, \citet[\S4.7]{HiranoToda2024JME}). As we discuss later, Proposition 1(c) of \citet{Tirole1985} has been reappraised only recently. Therefore, it is crucial to correctly understand a benchmark model with both a dividend-paying asset and capital accumulation such as \citet{Tirole1985}.

The idea of our counterexample (Proposition \ref{prop:C}) is as follows. We construct a standard production function $f(k)$ with $f'>0$ and $f''<0$ such that the wage $f(k)-kf'(k)$ is approximately linear in capital $k$ near $k=0$. Then if capital $k_t$ is small today, so is the wage. If future dividend $d_{t+1}$ is not too small, the asset price is not too small, so the budget constraint of the young (wage is spent on consumption, asset purchase, and future capital) forces future capital $k_{t+1}$ to be small. Thus, we may sustain an equilibrium in which $\set{k_t}$ converges to zero.\footnote{Thus, our counterexample shares a similarity with \citet{Chattopadhyay2008}, who provides a counterexample to the net dividend criterion of \citet{AbelMankiwSummersZeckhauser1989} for dynamic efficiency. However, our example is not directly related to \citet{Chattopadhyay2008}. We thank Gadi Barlevy for pointing out this issue.} 
Furthermore, we show that this is the unique equilibrium of the economy, the interest rate diverges to infinity, and there is no bubble (the asset price equals the present discounted value of dividends), although the model satisfies all assumptions of Proposition 1(c). This example can be understood through an analogy. During the 16th and 17th centuries, Spain experienced a significant influx of silver from its colonies in the Americas. Revenue from this wealth allowed Spain to import goods, which led to a decline in domestic manufacturing \citep{Drelichman2005}. This phenomenon is often referred to as the ``resource curse'' or the ``Dutch disease'' \citep{CordenNeary1982}. In our example, the asset initially pays large dividends relative to the capital stock but dividends decline over time; as a result, capital continues to decline.

Our counterexample has the feature that initial capital is so small that capital keeps declining, preventing the economy from converging to a steady state with positive capital. However, in Theorem \ref{thm:restore}, we restore Proposition 1(c) of \citet{Tirole1985} under the additional assumptions that initial capital is sufficiently large and dividends are sufficiently small. Example \ref{exmp:initial} shows that convergence to zero or a positive steady state is possible depending on the initial condition. Example \ref{exmp:D0} shows that the resource curse arises for any initial capital level by taking a sufficiently large initial dividend. Therefore, the assumptions of sufficiently large initial capital and sufficiently small dividends are essential for restoring Proposition 1(c). As the literature on rational bubbles tends to focus on the steady state and ignore the case $k=0$ due to the Inada condition, the fact that $k_t\to 0$ can robustly arise depending on the initial condition may be surprising.

\paragraph*{Related literature}

Proposition 1(c) of \citet{Tirole1985} concerns an environment in which an asset price bubble is the unique equilibrium outcome. This is a very important (though overlooked) insight, as asset price bubbles are often considered fragile and not robust \citep{SantosWoodford1997}. To our knowledge, \citet[\S7]{Wilson1981} was the first to point out such an example in an endowment economy.\footnote{We thank Herakles Polemarchakis for bringing our attention to \citet{Wilson1981}.} In the literature, as we document in the supplementary material \citep{PhamTodaCommentSuppl}, Proposition 1(c) of \citet{Tirole1985} has been referred to only a few times including \citet[p.~351]{Burke1996}, \citet*[Footnote 8]{AllenBarlevyGale2017},\footnote{A significantly revised version of \citet{AllenBarlevyGale2017} was published as \citet{AllenBarlevyGale2025}.} and several papers by Hirano and Toda. \citet{HiranoToda2025JPE} prove the necessity of bubbles (i.e., bubbles must emerge in every equilibrium) in modern macro-finance models including overlapping generations models and infinite-horizon models. Their \S V.A formally analyzed the \citet{Tirole1985} model in the special case with logarithmic utility, but the authors were unable to dispense with the assumption on an endogenous object, namely that capital is bounded away from zero. (Indeed, our counterexample features an equilibrium path in which capital converges to zero.) Our Theorem \ref{thm:restore} completely resolves this issue.

\section{Tirole (1985)'s model}

As \citet{Tirole1985}'s model is well known (see \citet[\S5.2]{BlanchardFischer1989} for a textbook treatment), our model description is brief. Time is discrete and denoted by $t=0,1,\dotsc$. There are overlapping generations of agents who live for two dates. Each agent is endowed with one unit of labor when young and none when old. Let $N_t=G^t$ be the population of generation $t$, where $G>0$ is the gross population growth rate. The utility function of generation $t$ is $U(c_t^y,c_{t+1}^o)$ , where $c_t^y,c_{t+1}^o$ denote consumption when young and old.

A representative firm produces the output using the neoclassical production function $F(K_t,L_t)$ (which includes undepreciated capital), where $K_t,L_t$ denote capital and labor inputs. Each agent supplies a unit of labor inelastically when young, so $L_t=N_t=G^t$ in equilibrium defined below. Let $k_t\coloneqq K_t/L_t=K_t/G^t$ be the capital per capita, $f(k)\coloneqq F(k,1)$, and assume $f'>0$, $f''<0$, $f'(0)=\infty$, and $f'(\infty)<G$.\footnote{\citet[p.~1501]{Tirole1985} explicitly states $f'(0)=\infty$ and implicitly assumes $f'>0>f''$. His Proposition 1(c) assumes the existence of a steady state $f'(k_b^*)=G$, which implies $f'(\infty)<G$.} The last condition rules out diverging paths (see the proof of Lemma \ref{lem:impossible2} below).

There is also a unit supply of an asset with infinite maturity. Let $D_t\ge 0$ be the (exogenous) dividend and $P_t\ge 0$ be the (endogenous) price. The young choose savings $s_t$ to maximize the lifetime utility. Given initial capital $K_0>0$, a \emph{perfect foresight equilibrium} consists of a sequence $\set{(P_t,R_{t+1},w_t,s_t,K_t)}_{t=0}^\infty$ of asset price, interest rate, wage, savings, and capital such that the following conditions hold:
\begin{subequations}
\begin{align}
    s_t&=\argmax_{s\in [0,w_t]}U(w_t-s,R_{t+1}s), \label{eq:eq_utility}\\
    (R_t,w_t)&=(f'(k_t),f(k_t)-k_tf'(k_t)), \label{eq:eq_profit}\\
    P_t&=\frac{1}{R_{t+1}}(P_{t+1}+D_{t+1}), \label{eq:eq_noarbitrage}\\
    N_ts_t&=K_{t+1}+P_t. \label{eq:eq_assetclear}
\end{align}
\end{subequations}
Here, condition \eqref{eq:eq_utility} is utility maximization; \eqref{eq:eq_profit} is the first-order condition for profit maximization; \eqref{eq:eq_noarbitrage} is the no-arbitrage condition between capital and asset; and \eqref{eq:eq_assetclear} is asset market clearing that equates aggregate savings (left-hand side) to the market capitalization of safe assets (right-hand side).

The \emph{fundamental value} of the asset is the present discounted value of dividends
\begin{equation}
    V_t\coloneqq \sum_{s=1}^\infty \frac{D_{t+s}}{R_{t+1}\dotsb R_{t+s}}. \label{eq:Vt}
\end{equation}
We say that the equilibrium is \emph{bubbleless} if $P_t=V_t$, and \emph{bubbly} if $P_t>V_t$. Furthermore, letting $p_t\coloneqq P_t/G^t$ be the detrended asset price, we say that the equilibrium is \emph{asymptotically bubbly} if $P_t>V_t$ and $\liminf_{t\to\infty}p_t>0$. It is convenient to define the long-run dividend growth rate by $G_d\coloneqq \limsup_{t\to\infty} D_t^{1/t}$
and the detrended dividend $d_t\coloneqq D_t/G^t$. See \citet{HiranoToda2025JPE} for more discussion of these concepts, especially their \S II and Definitions 1, 2.

\citet[p.~1502]{Tirole1985} imposes several assumptions on functions describing the equilibrium system. \citet[\S C.1]{PhamTodaWP} argue that we can justify these assumptions if the savings function $s(w,R)$ (the solution to the utility maximization problem \eqref{eq:eq_utility} given $(w_t,R_{t+1})=(w,R)$) is strictly increasing in $w$ and increasing in $R$. We can justify this assumption, in turn, if the utility function is additively separable as $U(c^y,c^o)=u(c^y)+v(c^o)$ and $v$ exhibits relative risk aversion bounded above by 1 \citep[Lemma 2.3]{PhamTodaWP}. Thus, we maintain the following assumption.

\begin{asmp}[Monotonicity of saving]\label{asmp:U}
The utility function $U$ is twice differentiable, strictly quasi-concave, satisfies the Inada condition, and the savings function $s(w,R)$ satisfies $s_w>0$ and $s_R\ge 0$.
\end{asmp}

Under the monotonicity condition on $s$, we obtain the following result, which is similar to Lemma 1 of \citet*{BosiHa-HuyLeVanPhamPham2018}.

\begin{lem}[Equilibrium system]\label{lem:system}
If Assumption \ref{asmp:U} holds, the equation
\begin{equation}
    Gx+p-s(f(k)-kf'(k),f'(x))=0 \label{eq:xeq}
\end{equation}
has at most one solution $x=g(k,p)>0$, which satisfies $g_k>0$ and $g_p<0$ on its domain. Letting $(k_t,p_t,d_t)=(K_t,P_t,D_t)/G^t$, given $k_0>0$, an equilibrium has a one-to-one correspondence with the system
\begin{subequations}\label{eq:system}
\begin{align}
    k_{t+1}&=g(k_t,p_t), \label{eq:system_k}\\
    p_{t+1}&=\frac{f'(k_{t+1})}{G}p_t-d_{t+1}.\label{eq:system_p}
\end{align}
\end{subequations}
\end{lem}

\begin{exmp}[Logarithmic utility]\label{exmp:log}
Consider the logarithmic utility
\begin{equation}
    U(c^y,c^o)=(1-\beta)\log c^y+\beta \log c^o, \label{eq:utility_log}
\end{equation}
where $\beta\in (0,1)$ governs time preference. Then the savings function is $s(w,R)=\beta w$, which satisfies Assumption \ref{asmp:U}. The function $g$ in \eqref{eq:system_k} reduces to
\begin{equation}
    g(k,p)=\frac{\beta(f(k)-kf'(k))-p}{G}, \label{eq:g_log}
\end{equation}
whose domain is $(k,p)\in \R_{++}\times \R_+$ such that $p<\beta(f(k)-kf'(k))$.
\end{exmp}

By Lemma \ref{lem:system}, an equilibrium has a one-to-one correspondence with a sequence $\set{(k_t,p_t)}_{t=0}^\infty$ satisfying \eqref{eq:system}. Noting that \eqref{eq:system} is recursive and $k_0>0$ is given, an equilibrium has a one-to-one correspondence with the initial asset price $p_0$. For this reason, in what follows we often say ``$\set{(k_t,p_t)}_{t=0}^\infty$ is an equilibrium'' or ``$p_0$ is an equilibrium''. Using Lemma \ref{lem:system}, we can show that the set of the initial asset price $p_0$ in equilibrium, denoted $\cP_0$, is an interval (possibly a singleton), and the equilibrium paths satisfy some monotonicity property. The following lemma is an adaptation of Lemmas 4, 6, 10 of \citet{Tirole1985} and hence we omit the proof. (See \citet[Proposition 2.2]{PhamTodaWP}.)

\begin{lem}[Equilibrium monotonicity]\label{lem:p0}
If Assumption \ref{asmp:U} holds, the equilibrium set $\cP_0$ is an interval. Let $p_0,p_0'\in \cP_0$ and $p_0<p_0'$. Let $\set{(k_t,p_t)}_{t=0}^\infty$ satisfy the equilibrium system \eqref{eq:system}, $R_t=f'(k_t)$, $w_t=f(k_t)-k_tf'(k_t)$, and let $p_t=v_t+b_t$ be the fundamental-bubble decomposition obtained by $p_t\coloneqq P_t/G^t$ and $v_t\coloneqq V_t/G^t$ in \eqref{eq:Vt}. Define $(k_t',p_t',R_t',w_t',v_t',b_t')$ analogously. Then for all $t\ge 1$ we have $k_t>k_t'$, $p_t<p_t'$, $R_t<R_t'$, $w_t>w_t'$, $v_t\ge v_t'$, and $b_t<b_t'$.
\end{lem}

The following uniqueness result plays a crucial role for constructing our counterexample. It states that if there exists an equilibrium with the long-run interest rate exceeding the population growth rate, then it is bubbleless, and there exist no other equilibria.

\begin{lem}[Unique, bubbleless equilibrium]\label{lem:impossible2}
Suppose Assumption \ref{asmp:U} holds. Let $\set{(k_t,p_t)}_{t=0}^\infty$ be an equilibrium and $\bar{k}\coloneqq \limsup_{t\to\infty}k_t$. If $f'(\bar{k})>G$, then the equilibrium is bubbleless and no other equilibrium (bubbly or bubbleless) exists.
\end{lem}

\section{Counterexample to Proposition 1(c)}

Let $\phi$ be an arbitrary positive, increasing, and concave function, and set the production function to $f(k)=A\phi(k)$, where $A>0$ is productivity. The wage is a rescaled version of $\omega(k)\coloneqq \phi(k)-k\phi'(k)>0$. The concavity of $\phi$ requires $\omega$ to be increasing. For deriving our counterexample, it is convenient if $\omega(k)$ is close to linear around $k=0$. Finally, we would like $\omega$ to be simple enough so that we can solve for $\phi(k)=k\int (\omega(x)/x^2)\diff x$ in closed-form. Setting $\omega(k)=k/(1+k)$ achieves all these requirements. Thus, define
\begin{equation}
    \phi(k)\coloneqq k\int_k^\infty \frac{1}{x(1+x)}\diff x=k\log (1+1/k), \label{eq:h}
\end{equation}
whose graph is shown in Figure \ref{fig:phi}.

\begin{figure}[htb!]
\centering
\begin{tikzpicture}
  \begin{axis}[
      xmax=6, ymax=1,
      axis lines = left,
      xlabel = \(k\),
      scale = 1.5,
    ]
    \addplot[domain=0:6, samples=1000, smooth, thick] {x*ln(1+1/x};
  \end{axis}
\end{tikzpicture}
\caption{The graph of $\phi(k)=k\log(1+1/k)$.}\label{fig:phi}
\end{figure}

Note that
\begin{subequations}\label{eq:hderiv}
\begin{align}
\phi'(k)&=\log(1+1/k)-\frac{1}{1+k}, \label{eq:hderiv1}\\
\phi''(k)&=\frac{1}{1+k}-\frac{1}{k}+\frac{1}{(1+k)^2}=-\frac{1}{k(1+k)^2}<0, \label{eq:hderiv2}\\
\phi(k)-k\phi'(k)&=\frac{k}{1+k}, \label{eq:hderiv3}
\end{align}
\end{subequations}
$\phi'(0)=\infty$, $\phi'(\infty)=\log 1=0$, and hence $\phi'(k)>0$ for $k<\infty$.

Since $f(k)=A\phi(k)$, by \eqref{eq:hderiv3} we obtain the wage
\begin{equation}
    f(k)-kf'(k)=A\frac{k}{1+k}. \label{eq:wage_h}
\end{equation}
For any utility function, since the savings function necessarily satisfies $s(w,R)\le w$, by Lemma \ref{lem:system} the equilibrium system satisfies
\begin{subequations}\label{eq:system1}
\begin{align}
    Gk_{t+1}+p_t&\le A\frac{k_t}{1+k_t}, \label{eq:system1_k}\\
    p_{t+1}+d_{t+1}&=\frac{f'(k_{t+1})}{G}p_t. \label{eq:system1_p}
\end{align}
\end{subequations}

The following lemma shows that, independent of preferences, if dividends grow at least geometrically fast, then equilibrium detrended capital converges to zero if initial capital is small enough.

\begin{lem}[Resource curse]\label{lem:curse}
Consider the production function $f(k)=A\phi(k)$ given by \eqref{eq:h}. Suppose dividends satisfy $D_t\ge DG_d^t$, where $D>0$ and $G_d\in (0,G)$. Let $r\coloneqq G_d/G\in (0,1)$ and $x_t=Cr^t/t>0$ for $t\ge 1$, where $C>0$. Then there exists $n\in \N$ such that if $k_0\le x_n$ and $\set{(k_t,p_t)}$ satisfies \eqref{eq:system1}, then $k_t\le x_{t+n}$ for all $t$ and $\lim_{t\to\infty}(k_t,p_t)=(0,0)$.
\end{lem}

Lemma \ref{lem:curse} is an example of the ``resource curse''.\footnote{To the best of our knowledge, \citet[Example 1]{BosiHa-HuyLeVanPhamPham2018} provide the first example of an equilibrium (in a model similar to that of \citet{Tirole1985} but with altruism and non-stationary dividends) where  $\lim_{t\to\infty}(k_t,p_t)=(0,0)$. They refer to this situation as the ``resource curse''. However, we can verify that this example satisfies $R>G$ and hence is not a counterexample to Proposition 1(c) of \citet{Tirole1985}.} Let us explain the intuition. By Lemma \ref{lem:system} and $s(w,R)\le w$, we obtain $Gk_{t+1}+p_t\le \omega(k_t)$, where $\omega(k)\coloneqq f(k)-kf'(k)>0$. Dividing both sides by $G>0$, using the no-arbitrage condition \eqref{eq:system_p}, and using $p_{t+1}\ge 0$, we obtain
\begin{equation}
    k_{t+1}+\frac{d_{t+1}}{f'(k_{t+1})}\le \frac{\omega(k_t)}{G}. \label{eq:kineq}
\end{equation}
Since $f'>0$, $f''<0$, and $f'(0)=\infty$, the left-hand side of \eqref{eq:kineq} is strictly increasing in $k_{t+1}$ and maps $(0,\infty)$ to $(0,\infty)$. Hence we may apply the implicit function theorem and rewrite \eqref{eq:kineq} as $k_{t+1}\le \psi(k_t,d_{t+1})$, where $\psi(k,d)$ is increasing in $k$ and decreasing in $d$. If $k_t$ is small, so is $k_{t+1}$ as long as $d_{t+1}$ is not too small. Hence, we may sustain an equilibrium in which $\set{k_t}$ converges to 0, which is the resource curse. We can now construct a counterexample to Proposition 1(c).

\setcounter{prop}{2}
\begin{prop}[Counterexample to \citet{Tirole1985}, Proposition 1(c)]\label{prop:C}
Suppose dividends satisfy $D_t\ge DG_d^t$, where $D>0$ and $G_d\in (0,G)$. Consider the logarithmic utility \eqref{eq:utility_log} and the production function $f(k)=A\phi(k)$ given by \eqref{eq:h}, where
\begin{equation}
    A\ge \max\set{2G/\beta,(G/\beta)^2/G_d}. \label{eq:AB}
\end{equation}
Then the following statements are true.
\begin{enumerate}
    \item\label{item:c1} The economy without the asset has a unique steady state $k^*=\beta A/G-1> 0$, which has steady-state interest rate $f'(k^*)<G_d$.
    \item\label{item:c2} There exists $\kappa>0$ such that, if $k_0<\kappa$, then the economy has a unique equilibrium $\set{(k_t,p_t)}_{t=0}^\infty$, which is bubbleless and converges to $(0,0)$.
\end{enumerate}
\end{prop}

In \citet{Tirole1985}, dividends are $D_t=D$ (constant) and $G>1$, so his model satisfies the assumptions of Proposition \ref{prop:C} with $G_d=1$ (along with all other assumptions on the utility function, production function, steady state, etc.). The reader may wonder where Tirole's proof went wrong. In \citet{Tirole1985}, the possibility of a bubbleless equilibrium with $R<G$ is considered at the bottom of p.~1522, where he states ``Let us now show that if $\bar{r}<0$ [corresponding to $R<G_d<G$], there exists no [bubbleless] equilibrium''. Here, Tirole states ``Let us consider the three mutually
exhaustive cases'', which are (in our notation)
\begin{enumerate*}
    \item\label{item:Rdec1} $R_t<R_{t-1}$ and $R_t<G$ for some $t$,
    \item\label{item:Rdec2} $R_t<R_{t-1}$ for some $t$, and $R_t\ge G$ for any such $t$, and
    \item\label{item:Rinc} $R_t\ge R_{t-1}$ for all $t$.
\end{enumerate*}
However, in each case, Tirole reasons that if the asset price converges to 0, the interest rate must converge to the bubbleless interest rate. This reasoning is incorrect, as our counterexample satisfies $p_t\to 0$ yet $R_t=f'(k_t)\to \infty$ (because $k_t\to 0$).

\section{Restoring Proposition 1(c)}

Since a counterexample exists, \citet{Tirole1985}'s original claim in Proposition 1(c) cannot be true without additional assumptions. Notice that to construct an equilibrium with $k_t\to 0$ (``resource curse''), Lemma \ref{lem:curse} requires the initial capital to be sufficiently small. We can thus conjecture that if initial capital is sufficiently large, the conclusion of Proposition 1(c) may be true. In this section, we show that this is indeed the case, provided that dividends are sufficiently small. To this end, we introduce an additional assumption.

\begin{asmp}[Bubbly steady state]\label{asmp:g}
Let $g$ be as in Lemma \ref{lem:system}. There exist $k^*,p^*>0$ such that $k^*=g(k^*,p^*)$ and $f'(k^*)=G$.
\end{asmp}

Assumption \ref{asmp:g} merely implies that $(k^*,p^*)$ is a bubbly steady state. Note that $k^*$ is unique because $f''<0$. Then $p^*$ is also unique because $g$ is strictly decreasing in $p$ by Lemma \ref{lem:system}. Furthermore, let
\begin{equation}
    \cK^*\coloneqq \set{k>0:k=g(k,0)} \label{eq:cK}
\end{equation}
be the set of steady-state capital without the asset (which could be empty). The following theorem shows that if 
\begin{enumerate*}
    \item the bubbleless interest rate is less than the dividend growth rate,
    \item initial capital is sufficiently large (not too small relative to the bubbly steady state value), and
    \item dividends are sufficiently small, then there exists a unique equilibrium, which is asymptotically bubbly.
\end{enumerate*}
Thus, we restore Proposition 1(c) of \citet{Tirole1985}.

\begin{thm}[Bubble necessity with small dividends]\label{thm:restore}
Suppose Assumptions \ref{asmp:U}, \ref{asmp:g} hold, $G_d\coloneqq \limsup_{t\to \infty}D_t^{1/t}\in (0,G)$, and $f'(k)<G_d$ for all $k\in \cK^*$ in \eqref{eq:cK}. Let $d_t\coloneqq D_t/G^t$ be the detrended dividend. Then there exist $\kappa\in (0,k^*)$ and $\delta>0$ such that, if $k_0\ge \kappa$ and $\sup_{t\ge 1} d_t\le \delta$, then there exists a unique equilibrium $\set{(k_t,p_t)}$, which is asymptotically bubbly and converges to $(k^*,p^*)$.
\end{thm}

\citet[p.~1502]{Tirole1985} assumes $\cK^*$ in \eqref{eq:cK} is a singleton, which he refers to as ``Diamond's stability assumption''; we do not require it. The condition $f'(k)<G_d<G$ for all $k\in \cK^*$ corresponds to the ``bubble necessity condition'' $R<G_d<G$ in \citet{HiranoToda2025JPE}. The key to rectifying Proposition 1(c) is to choose initial capital not too small and dividends not too large. 

The following example illustrates the importance of the initial condition.

\begin{exmp}[Importance of initial condition]\label{exmp:initial}
Consider the economy in Proposition \ref{prop:C} with $D_t=DG_d^t$, where $D>0$. Then Assumption \ref{asmp:U} and the conditions on dividends hold. The steady state condition $k=g(k,p)$ is equivalent to
\begin{equation}
    Gk+p=\beta A\frac{k}{1+k}\iff G+\frac{p}{k}=\frac{\beta A}{1+k}. \label{eq:ss}
\end{equation}
Let $k_f^*,k_b^*$ be the fundamental and bubbly steady-state capital. By Proposition \ref{prop:C}, we have $k_f^*=\beta A/G-1>0$ and $f'(k_f^*)<G_d$. Let $k_b^*<k_f^*$ be such that $f'(k_b^*)=G$. Since the right-hand side of \eqref{eq:ss} is decreasing in $k$, we must have $p/k>0$ at $k=k_b^*$ and Assumption \ref{asmp:g} holds. Since the set $\cK^*=\{k_f^*\}$ is a singleton, by Theorem 3 of \citet{PhamTodaWP}, there exists a unique equilibrium, and either $k_t\to 0$ or $k_t\to k_b^*$.\footnote{Note that Lemma \ref{lem:unique2} only claims that there exists a unique equilibrium and either $p_t\to 0$ or $p_t\to p^*$, which is not enough for this argument.} Let $\cK_0\coloneqq \set{k_0>0:k_t\to k_b^*}$ and $\kappa\coloneqq \inf \cK_0$. By the definition of $\cK_0$ and Theorem \ref{thm:restore}, if $D>0$ is small enough, it must be $\kappa<k_b^*$ and $(\kappa,\infty)\subset \cK_0\subset [\kappa,\infty)$. By Proposition \ref{prop:C}, $(0,\infty)\backslash \cK_0$ is nonempty, so it must be $\kappa>0$. Therefore, $k_t\to 0$ if $k_0<\kappa$ and $k_t\to k_b^*$ if $k_0>\kappa$.
\end{exmp}

The following example shows that the assumption of sufficiently small dividends in Theorem \ref{thm:restore} is essential.

\begin{exmp}[Large dividends imply resource curse]\label{exmp:D0}
Consider the same economy as Example \ref{exmp:initial}. For any $k_0>0$ and $G_d\in (0,G)$, let $r\coloneqq G_d/G\in (0,1)$. Choose $D>0$ large enough such that \eqref{eq:Dineq} holds for all $m$, where we set $n=1$ and $C=k_0/r$. Then by Lemma \ref{lem:curse}, we have $k_t\le k_0r^t/(t+1)$ for all $t$, so $k_t\to 0$.
\end{exmp}

Finally, the following theorem shows that, under additional Inada-type conditions, the conclusion of Theorem \ref{thm:restore} holds for arbitrary dividends.

\begin{thm}[Bubble necessity with arbitrary dividends]\label{thm:restore2}
Let everything be as in Theorem \ref{thm:restore}. Suppose that $k\mapsto f(k)-kf'(k)$ has range $(0,\infty)$ and for any fixed $c^o>0$, we have
\begin{equation}
    \lim_{c^y\to \infty}\frac{U_1}{U_2}(c^y,c^o)=0. \label{eq:sufficient}
\end{equation}
Then there exists $\kappa>0$ such that for all $k_0\ge \kappa$, the conclusion of Theorem \ref{thm:restore} holds.
\end{thm}

Proposition \ref{prop:C} and Theorem \ref{thm:restore} show that resource curse and bubble necessity are both theoretically possible. Each case seems to be empirically relevant based on the anecdotal evidence of the Spanish colonization of the Americas cited in the introduction and the connection between unbalanced growth (caused by technological innovation) and asset price bubbles \citep{HiranoToda2025PNAS}.

\appendix

\section{Proofs}\label{sec:proof}

\subsection{Proof of Lemma \ref{lem:system}}

Let $\Phi(x,k,p)$ be the left-hand side of \eqref{eq:xeq}. Then
\begin{align*}
    \Phi_x&=G-s_Rf''(x)>0,& \Phi_k&=s_wkf''(k)<0, & \Phi_p&=1.
\end{align*}
Therefore, $x=g(k,p)$ is unique (if it exists). By the implicit function theorem, we have
\begin{align*}
    g_k&=-\frac{\Phi_k}{\Phi_x}=-\frac{s_wkf''(k)}{G-s_Rf''(x)}>0,\\
    g_p&=-\frac{\Phi_p}{\Phi_x}=-\frac{1}{G-s_Rf''(x)}<0.
\end{align*}
\eqref{eq:system_k} follows from dividing the asset market clearing condition \eqref{eq:eq_assetclear} by $N_t=G^t$ and noting that it is equivalent to setting $(k,p,x)=(k_t,p_t,k_{t+1})$ in \eqref{eq:xeq}. \eqref{eq:system_p} follows by dividing the no-arbitrage condition \eqref{eq:eq_noarbitrage} by $G^t$ and rearranging. \hfill \qedsymbol

\subsection{Proof of Lemma \ref{lem:impossible2}}

Take any equilibrium. We first show $\set{(k_t,p_t)}$ is uniformly bounded. Dividing \eqref{eq:eq_assetclear} by $N_t=G^t$ and noting $P_t\ge 0$ and $s_t\le w_t\le f(k_t)$, we obtain $Gk_{t+1}\le f(k_t)$. Since $F$ is neoclassical, $f(k)=F(k,1)$ is concave. Since $f'(\infty)<G$, we can take constants $a\in (0,1)$ and $b\ge 0$ such that $f(k)/G\le ak+b$ for all $k>0$. Iterating $0\le k_{t+1}\le f(k_t)/G\le ak_t+b$ yields
\begin{equation*}
    k_t\le a^t\left(k_0-\frac{b}{1-a}\right)+\frac{b}{1-a}.
\end{equation*}
Letting $t\to\infty$, we obtain $\limsup_{t\to\infty}k_t\le b/(1-a)$, so $\set{k_t}$ is uniformly bounded. Similarly, \eqref{eq:eq_assetclear} yields $p_t\le s_t\le w_t\le f(k_t)$, so $\set{p_t}$ is uniformly bounded.

Take $p>0$ such that $p_t\le p$ for all $t$. Since $\bar{k}=\limsup_{t\to\infty}k_t$ and $f'(\bar{k})>G$, we can take $\epsilon>0$ and $T>0$ such that $f'(\bar{k}+\epsilon)>G$ and $k_t<\bar{k}+\epsilon$ for all $t\ge T$. Let $R_t=f'(k_t)$. Then for $t>T$, we have
\begin{equation}
    \frac{P_t}{R_1\dotsb R_t}\le \frac{pG^t}{R_1\dotsb R_T f(\bar{k}+\epsilon)^{t-T}}=\frac{pG^T}{R_1\dotsb R_T}\left(\frac{G}{f'(\bar{k}+\epsilon)}\right)^{t-T}\to 0 \label{eq:nobubble_ub}
\end{equation}
as $t\to\infty$, so there is no bubble. (See Equation (5) of \citet{HiranoToda2025JPE}.) Suppose there exists another equilibrium $\set{(k_t',p_t')}_{t=0}^\infty$. Let $b_t,b_t'\ge 0$ be the bubble components of these equilibria. Since $\set{(k_t,p_t)}_{t=0}^\infty$ is bubbleless, we have $b_t=0$. By Lemma \ref{lem:p0}, $0\le b_t'\neq b_t=0$ implies $b_t'>b_t$ and hence $R_t'>R_t$. By the same derivation as \eqref{eq:nobubble_ub}, it follows that $\set{(k_t',p_t')}_{t=0}^\infty$ is bubbleless, which contradicts $b_t'>0$. Therefore, the equilibrium is unique, and it is bubbleless. \hfill \qedsymbol

\subsection{Proof of Lemma \ref{lem:curse}}

Let $d_t\coloneqq D_t/G^t$ be the detrended dividend, which satisfies $d_t\ge Dr^t$. We seek to prove the claim by induction. By \eqref{eq:system1_p} and $f(k)=A\phi(k)$, we obtain
\begin{equation}
    \frac{p_t}{G}=\frac{p_{t+1}+d_{t+1}}{f'(k_{t+1})}>\frac{d_{t+1}}{f'(k_{t+1})}=\frac{d_{t+1}}{A\phi'(k_{t+1})}. \label{eq:pt_ub}
\end{equation}
By \eqref{eq:system1_k} and \eqref{eq:pt_ub}, if $k_t\le x_{t+n}$, then
\begin{equation}
    k_{t+1}+\frac{d_{t+1}}{A\phi'(k_{t+1})}<\frac{A}{G}\frac{k_t}{1+k_t}\le \frac{A}{G}\frac{x_{t+n}}{1+x_{t+n}}. \label{eq:ineq1}
\end{equation}
Noting that $x\mapsto x+d_{t+1}/(A\phi'(x))$ is strictly increasing (because $\phi''<0$), if we can show
\begin{equation}
    \frac{A}{G}\frac{x_{t+n}}{1+x_{t+n}}\le x_{t+n+1}+\frac{d_{t+1}}{A\phi'(x_{t+n+1})}, \label{eq:ineq2}
\end{equation}
then $k_{t+1}\le x_{t+n+1}$ follows from \eqref{eq:ineq1} and \eqref{eq:ineq2}. Therefore, it suffices to show \eqref{eq:ineq2}. But noting that $x_{t+n}>0$ and $d_{t+1}\ge Dr^{t+1}$, it suffices to show
\begin{equation}
    \frac{A}{G}x_{t+n}\le x_{t+n+1}+\frac{Dr^{t+1}}{A\phi'(x_{t+n+1})}. \label{eq:ineq3}
\end{equation}
Now set $x_t=Cr^t/t$, where $C>0$. Using \eqref{eq:hderiv1}, \eqref{eq:ineq3} is equivalent to
\begin{equation}
    \frac{A}{G}\frac{Cr^{t+n}}{t+n}\le \frac{Cr^{t+n+1}}{t+n+1}+\frac{D}{A}\frac{r^{t+1}}{\log\left(1+\frac{t+n+1}{Cr^{t+n+1}}\right)-\frac{1}{1+\tfrac{Cr^{t+n+1}}{t+n+1}}}. \label{eq:ineq4}
\end{equation}
Setting $m=t+n+1\ge 2$ and noting that $\phi'>0$, \eqref{eq:ineq4} is equivalent to
\begin{equation}
    \frac{D}{r^n}\ge AC\left(\frac{A}{Gr}\frac{m}{m-1}-1\right)\left(\frac{\log(1+m/(Cr^m))}{m}-\frac{1}{m+Cr^m}\right)\eqqcolon E_m. \label{eq:Dineq}
\end{equation}
A straightforward calculation shows
\begin{equation*}
    \lim_{m\to\infty}E_m=AC\left(\frac{A}{Gr}-1\right)(-\log r),
\end{equation*}
which is finite. Since $D>0$ and $r\in (0,1)$, we can take $n\in \N$ large enough such that $D/r^n\ge \sup_{m\ge 2} E_m$. Then \eqref{eq:Dineq} holds for all $m\ge 2$.

Let $k_0\le x_n$ and $\set{(k_t,p_t)}$ satisfy the equilibrium system \eqref{eq:system1}. Let us prove by induction that $k_t\le x_{t+n}$ for all $t$. The claim holds for $t=0$ by assumption. Suppose the claim holds for some $t$ and consider $t+1$. Since \eqref{eq:Dineq} holds for all $m\ge 2$, so does \eqref{eq:ineq4} for all $t\ge 0$. Hence \eqref{eq:ineq3} holds, which implies $k_{t+1}\le x_{t+n+1}$. Since $x_t=Cr^t/t\to 0$, we have $k_t\to 0$. Then \eqref{eq:system1_k} implies $p_t\to 0$. \hfill \qedsymbol

\subsection{Proof of Proposition \ref{prop:C}}

We need the following lemma to prove Proposition \ref{prop:C}.

\begin{lem}\label{lem:log}
For $0<z\le 1/2$, we have $\log(1-z)>-z-z^2$.
\end{lem}

\begin{proof}
Let $f(z)=\log(1-z)+z+z^2$. Then
\begin{align*}
    f'(z)&=-\frac{1}{1-z}+1+2z, & f''(z)&=-\frac{1}{(1-z)^2}+2.
\end{align*}
Hence $f''(z)>0$ for $z<a\coloneqq 1-1/\sqrt{2}$ and $f''(z)\le 0$ for $z\in [a,1)$. The strict convexity of $f$ for $z<a$ and $f(0)=f'(0)=0$ imply $f(z)>0$ for $z\in (0,a]$. The concavity of $f$ for $z\in [a,1)$ and $f(a)>0$, $f(1/2)=-\log 2+3/4>0$ imply $f(z)>0$ for $z\in [a,1/2]$.
\end{proof}

\begin{proof}[Proof of Proposition \ref{prop:C}]
\ref{item:c1} Solving $k=g(k,0)$ in \eqref{eq:g_log} and using \eqref{eq:wage_h}, we obtain the steady-state capital without the asset
\begin{equation*}
    k=\frac{\beta A}{G}\frac{k}{1+k}\iff k^*=\frac{\beta A}{G}-1> 0,
\end{equation*}
where we use $A\ge 2G/\beta$ in \eqref{eq:AB}. Using \eqref{eq:hderiv1}, the steady-state interest rate is
\begin{equation}
    R^*\coloneqq f'(k^*)=A\phi'(k^*)=-A\log\left(1-\frac{G}{\beta A}\right)-\frac{G}{\beta}. \label{eq:R*}
\end{equation}
Setting $z=G/(\beta A)\le 1/2$ in \eqref{eq:R*} and applying Lemma \ref{lem:log}, we obtain
\begin{equation*}
    0<R^*=\frac{G}{\beta}\left(-\frac{\log(1-z)}{z}-1\right)<\frac{G}{\beta}z=\left(\frac{G}{\beta}\right)^2\frac{1}{A}\le G_d,
\end{equation*}
where the last inequality follows from $A\ge (G/\beta)^2/G_d$ in \eqref{eq:AB}.

\medskip
\noindent
\ref{item:c2} By Theorem 1 of \citet{PhamTodaWP}, an equilibrium $\set{(k_t,p_t)}$ exists. Let $r\coloneqq G_d/G\in (0,1)$, $x_t\coloneqq Cr^t/t>0$ for $t\ge 1$ and $C>0$, and choose $n\in \N$ as in Lemma \ref{lem:curse}. If $k_0\le \kappa \coloneqq x_n$, then $(k_t,p_t)\to (0,0)$ by Lemma \ref{lem:curse}. Since $f'(0)=\infty>G$, by Lemma \ref{lem:impossible2}, the equilibrium is unique and bubbleless.
\end{proof}

\subsection{Proof of Theorem \ref{thm:restore}}

We need several lemmas to prove Theorem \ref{thm:restore}. In what follows, we always assume $G_d\coloneqq \limsup_{t\to\infty}D_t^{1/t}<G$.

\begin{lem}[Long-run behavior of equilibrium]\label{lem:longrun}
If Assumption \ref{asmp:U} holds, in any equilibrium, one of the following statements is true.
\begin{enumerate}[(a)]
    \item\label{item:lr_bubbleless} The equilibrium is bubbleless, $\lim_{t\to\infty}p_t=0$, and $R_t>G$ for sufficiently large $t$.
    \item\label{item:lr_asymbubbleless} The equilibrium is asymptotically bubbleless and $\set{(k_t,p_t,R_t)}$ converges to $(k,0,R)$ satisfying $k=g(k,0)$ and $R=f'(k)\in [G_d,G]$.
    \item\label{item:lr_bubbly} The equilibrium is asymptotically bubbly and $\set{(k_t,p_t,R_t)}$ converges to $(k,p,G)$ satisfying $k=g(k,p)$, $p>0$, and $G=f'(k)$.
\end{enumerate}
\end{lem}

\begin{proof}
We omit the proof as it is essentially the same as Lemmas 2 and 3 of \citet{Tirole1985}. See \citet[Proposition 3.3]{PhamTodaWP}.
\end{proof}

The following lemma is a straightforward consequence of Lemmas \ref{lem:p0} and \ref{lem:longrun}.

\begin{lem}[Uniqueness of bubbleless and asymptotically bubbly equilibria]\label{lem:unique1}
If Assumption \ref{asmp:U} holds, bubbleless and asymptotically bubbly equilibria are unique.
\end{lem}

\begin{proof}
If $p_0<p_0'$ are two bubbleless equilibria, by Lemma \ref{lem:p0}, the bubble components satisfy $0=b_0<b_0'=0$, which is a contradiction.

If $p_0<p_0'$ are two asymptotically bubbly equilibria, by Lemma \ref{lem:p0}, we have $k_t>k_t'$, $0<p_t<p_t'$, and $0<R_t<R_t'$ for all $t\ge 1$. By Lemma \ref{lem:longrun}, $\set{(k_t,p_t,R_t)}$ and $\set{(k_t',p_t',R_t')}$ converge to $(k^*,p^*,G)$. Therefore, $\lim_{t\to\infty}p_t'/p_t=p^*/p^*=1$. However, $0<p_t<p_t'$, $0<R_t<R_t'$, and \eqref{eq:system_p} imply
\begin{equation*}
    \frac{p_t'}{p_t}=\frac{(R_t'/G)p_{t-1}'-d_t}{(R_t/G)p_{t-1}-d_t}\ge \frac{(R_t'/G)p_{t-1}'}{(R_t/G)p_{t-1}}>\frac{p_{t-1}'}{p_{t-1}},
\end{equation*}
so by induction $p_t'/p_t>\dots>p_0'/p_0>1$. Therefore, $\lim_{t\to\infty}p_t'/p_t\ge p_0'/p_0>1$, which is a contradiction.
\end{proof}

The following lemma establishes the uniqueness of equilibrium.

\begin{lem}[Uniqueness of equilibrium]\label{lem:unique2}
If Assumption \ref{asmp:U} holds and $f'(k)<G_d$ for all $k\in \cK^*$ in \eqref{eq:cK}, then there exists a unique equilibrium, which takes the form of either \ref{item:lr_bubbleless} or \ref{item:lr_bubbly} in Lemma \ref{lem:longrun}.
\end{lem}

\begin{proof}
By Theorem 1 of \citet{PhamTodaWP}, there exists an equilibrium. Note that Lemma \ref{lem:longrun} covers all cases regarding the behavior of $\set{R_t}$. Since $f'(k)<G_d$ for all $k\in \cK^*$, case \ref{item:lr_asymbubbleless} in Lemma \ref{lem:longrun} cannot occur. Hence, every equilibrium is either bubbleless or asymptotically bubbly.

To show equilibrium uniqueness, suppose $p_0,p_0'\in \cP_0$ and $p_0<p_0'$. By Lemma \ref{lem:p0}, $p_0'$ is bubbly, so it is asymptotically bubbly. It is also unique by Lemma \ref{lem:unique1}. Hence $p_0$ must be bubbleless, which is also unique by Lemma \ref{lem:unique1}. Therefore, the equilibrium set is the two-point set $\cP_0=\set{p_0,p_0'}$, which is a contradiction because Lemma \ref{lem:p0} implies that $\cP_0$ is an interval.
\end{proof}

The following lemma shows that, once we have an equilibrium converging to the bubbly steady state, increasing initial capital retains this property.

\begin{lem}\label{lem:k0}
Let everything be as in Lemma \ref{lem:unique2} and suppose an equilibrium $\set{(k_t,p_t)}$ of the form of Lemma \ref{lem:longrun}\ref{item:lr_bubbly} exists. If $k_0'>k_0$, the corresponding (unique) equilibrium $\set{(k_t',p_t')}$ is also of the form of Lemma \ref{lem:longrun}\ref{item:lr_bubbly}.
\end{lem}

\begin{proof}
By Lemma \ref{lem:unique2}, a unique equilibrium exists, which takes the form of either \ref{item:lr_bubbleless} or \ref{item:lr_bubbly} in Lemma \ref{lem:longrun}. Let $\set{(k_t',p_t')}$ be the corresponding equilibrium path.

We claim $p_0'>p_0$. Suppose to the contrary that $p_0'\le p_0$. As in Lemma \ref{lem:p0}, we can easily show $k_t'>k_t$, $p_t'<p_t$, and $R_t'<R_t$ for all $t\ge 1$. If the equilibrium is of the form of Lemma \ref{lem:longrun}\ref{item:lr_bubbleless}, then $p_t'\to 0$ and for large enough $t$ we have $G<R_t'<R_t\to G$, so $R_t'\to G$ and $k_t'\to k^*$. This forces $p_t'\to p^*$, which contradicts $p_t'\to 0$. Therefore the equilibrium is of the form of Lemma \ref{lem:longrun}\ref{item:lr_bubbly}, and $(k_t',p_t',R_t')\to (k^*,p^*,G)$. But then \eqref{eq:system_p} implies
\begin{equation*}
    \frac{p_t'}{p_t}=\frac{(R_t'/G)p_{t-1}'-d_t}{(R_t/G)p_{t-1}-d_t}\le \frac{(R_t'/G)p_{t-1}'}{(R_t/G)p_{t-1}}<\frac{p_{t-1}'}{p_{t-1}},
\end{equation*}
so by induction $p_t'/p_t<\dots<p_1'/p_1<1$. Therefore, $\lim_{t\to\infty}p_t'/p_t\le p_1'/p_1<1$, which contradicts $p_t'/p_t\to p^*/p^*=1$. Thus, $p_0'>p_0$.

Finally, we claim that $k_t'>k_t$ and $p_t'>p_t$ for all $t$, which implies that $\liminf_{t\to\infty}p_t'\ge \lim_{t\to\infty}p_t=p^*>0$ and hence the equilibrium $\set{(k_t',p_t')}$ takes the form of Lemma \ref{lem:longrun}\ref{item:lr_bubbly}. The claim holds for $t=0$. Suppose it holds until some $t$, and consider $t+1$. If $k_{t+1}'\le k_{t+1}$, we have $R_{t+1}'\ge R_{t+1}$. Using \eqref{eq:system_p} and $p_t'>p_t$, we obtain
\begin{equation*}
    \frac{p_{t+1}'}{p_{t+1}}=\frac{(R_{t+1}'/G)p_{t}'-d_{t+1}}{(R_{t+1}/G)p_{t}-d_{t+1}}\ge \frac{(R_{t+1}'/G)p_{t}'}{(R_{t+1}/G)p_{t}}\ge \frac{p_t'}{p_t}>1,
\end{equation*}
so $p_{t+1}'>p_{t+1}$. Then, we have $k'_{t+2}=g(k'_{t+1},p'_{t+1})\le g(k_{t+1},p_{t+1})=k_{t+2}$. By induction, we get $k_{t+s}'\le k_{t+s}$ and $p_{t+s}'/p_{t+s}>\dots >p_t'/p_t>1$ for all $s\ge 1$. Then $\liminf_{t\to\infty}p_t'>\lim_{t\to\infty}p_t=p^*$, which contradicts the fact that either $p_t'\to 0$ or $p_t'\to p^*$. Therefore, we have $k_{t+1}'> k_{t+1}$. Then, by using the same argument as the proof of $p_0'>p_0$, we have $p_{t+1}'>p_{t+1}$, and by induction, the claim is true for all $t$.
\end{proof}

The following lemma allows us to apply the implicit function theorem.

\begin{lem}\label{lem:lindiff}
Let $A$ be a real $2\times 2$ matrix with two real eigenvalues $\lambda_1,\lambda_2$ satisfying $\abs{\lambda_1}<1<\abs{\lambda_2}$; $\set{u_t}_{t=0}^\infty$ a bounded sequence in $\R^2$; $b=(b_1,b_2)\neq 0$ a row vector; and $c\in \R$. Then the system of equations
\begin{equation}
    x_{t+1}=Ax_t+u_t \label{eq:lindiff}
\end{equation}
with the initial condition $bx_0=c$ has a unique bounded solution $\set{x_t}_{t=0}^\infty$ in $\R^2$ if and only if the first entry of the row vector $bP$ is nonzero ($(bP)_1\neq 0$), where $P$ is the real invertible matrix that diagonalizes $A$:
\begin{equation}
    P^{-1}AP=\begin{bmatrix}
        \lambda_1 & 0 \\
        0 & \lambda_2
    \end{bmatrix}. \label{eq:diagonal}
\end{equation}
\end{lem}
\begin{proof}
Multiplying $P^{-1}$ from left to \eqref{eq:lindiff}, we obtain
\begin{equation}
    P^{-1}x_{t+1}=(P^{-1}AP)P^{-1}x_t+P^{-1}u_t. \label{eq:lindiff1}
\end{equation}
Letting $y_t=P^{-1}x_t$, $v_t=P^{-1}u_t$, and writing \eqref{eq:lindiff1} entry-wise, we obtain
\begin{subequations}
\begin{align}
    y_{1,t+1}&=\lambda_1y_{1,t}+v_{1,t}, \label{eq:lindiff1a}\\
    y_{2,t+1}&=\lambda_2y_{2,t}+v_{2,t}. \label{eq:lindiff1b}
\end{align}
\end{subequations}
If $\set{x_t},\set{u_t}$ are bounded, so are $\set{y_t},\set{v_t}$. Noting that $\abs{\lambda_2}>1$ and solving \eqref{eq:lindiff1b} forward, we can uniquely determine $\set{y_{2,t}}$ as
\begin{equation*}
    y_{2,t}=-\sum_{s=0}^\infty \lambda_2^{-s-1}v_{2,t+s},
\end{equation*}
which is bounded. Noting that $\abs{\lambda_1}<1$ and solving \eqref{eq:lindiff1a} backward, we can uniquely determine $\set{y_{1,t}}$ as a function of $y_{1,0}$,
\begin{equation*}
    y_{1,t}=\lambda_1^ty_{1,0}+\sum_{s=1}^t \lambda_1^{s-1} v_{t-s},
\end{equation*}
which is bounded. To determine $y_{1,0}$, we use the initial condition
\begin{equation*}
    c=bx_0=bPy_0=(bP)_1y_{1,0}+(bP)_2y_{2,0}.
\end{equation*}
Since $y_{2,0}$ is determined, $y_{1,0}$ is uniquely determined if and only if $(bP)_1\neq 0$.
\end{proof}

\begin{proof}[Proof of Theorem \ref{thm:restore}]
By Lemma \ref{lem:unique2}, there exists a unique equilibrium, which takes the form of either \ref{item:lr_bubbleless} or \ref{item:lr_bubbly} in Lemma \ref{lem:longrun}. If we can show that an asymptotically bubbly equilibrium exists if $k_0>0$ is sufficiently close to $k^*$, then by Lemma \ref{lem:k0} the claim holds for all $k_0>\kappa$ for some $\kappa\in (0,k^*)$. Therefore, it suffices to show the existence of an equilibrium of the form of Lemma \ref{lem:longrun}\ref{item:lr_bubbly} when $k_0>0$ is sufficiently close to $k^*$ and detrended dividends $\set{d_t}$ are sufficiently small.

We prove this claim by applying the implicit function theorem. The proof uses functional analysis and we refer the reader to \citet{Luenberger1969}. When $k_0=k^*$ and $d_t=0$ for all $t$ (stationary pure bubble model), such an equilibrium trivially exists, namely $(k_t,p_t)=(k^*,p^*)$ for all $t$. Now consider the case with general $k_0$ and $\set{d_t}$. By Lemma \ref{lem:system}, the equilibrium system is described by \eqref{eq:system}. By Assumption, we have $\limsup_{t\to\infty} d_t^{1/t}<1$, so in particular $d_t\to 0$ and $\set{d_t}$ is bounded. Let $x_0=k_0$, $x_t=d_t$ for $t\ge 1$, and $x=(x_t)\in \ell^\infty\eqqcolon X$, where $\ell^\infty$ denotes the Banach space of real bounded sequences equipped with the supremum norm $\norm{\cdot}$. Let $y_t=(k_{t+1},p_t)$ and $y=(y_t)\in (\ell^\infty)^2\eqqcolon Y$, which is also a Banach space with the supremum norm. We say $y$ is positive and write $y>0$ if $k_{t+1}>0$ and $p_t>0$ for all $t$. Let $Z\coloneqq Y=(\ell^\infty)^2$. Define the operator $\Phi:X\times Y\to Z$ by
\begin{equation}
    \Phi(x,y)=(\Phi_0(x,y),\dots,\Phi_t(x,y),\dotsc), \label{eq:Phi}
\end{equation}
where we restrict $y>0$ and
\begin{equation*}
    \Phi_t(x,y)=\begin{bmatrix}
        k_{t+1}-g(k_t,p_t) \\
        p_{t+1}-\frac{f'(k_{t+1})}{G}p_t+d_{t+1}
    \end{bmatrix}.
\end{equation*}
Then $\Phi$ is continuously Fr\'echet differentiable.  Letting $D_y\Phi$ denote the Fr\'echet derivative with respect to $y$, we may view $D_y\Phi$ as a block matrix whose $(t,j)$ block is
\begin{equation}
    D_{y_j}\Phi_t(x,y)=\begin{cases*}
        \begin{bmatrix}
            -g_k(k_t,p_t) & 0\\
            0 & 0
        \end{bmatrix} & if $j=t-1$,\\
        \begin{bmatrix}
            1 & -g_p(k_t,p_t) \\
            -\frac{f''(k_{t+1})}{G}p_t & -\frac{f'(k_{t+1})}{G}
        \end{bmatrix} & if $j=t$,\\
        \begin{bmatrix}
            0 & 0\\
            0 & 1
        \end{bmatrix} & if $j=t+1$,\\
        0 & otherwise.
    \end{cases*}\label{eq:DyPhi}
\end{equation}

To apply the implicit function theorem, let $x^*,y^*$ be the $x,y$ corresponding to the steady state, namely $x^*=(k^*,0,0,\dotsc)$ and $y^*=\set{(k^*,p^*)}$. We evaluate $D_y\Phi$ at $(x^*,y^*)$. Since the entries of \eqref{eq:DyPhi} are constant at $(x^*,y^*)$, clearly $D_y\Phi(x^*,y^*):Y\to Z$ is a bounded linear operator. Let us show that $D_y\Phi(x^*,y^*)$ is bijective. To this end, consider the equation $z=D_y\Phi(x^*,y^*)h$, where $z=(z_t)$, $z_t=(z_{1,t},z_{2,t})$, and similarly for $h$. Decomposing the equation into blocks using \eqref{eq:DyPhi}, we obtain
\begin{align*}
    z_0&=D_{y_0}\Phi_0h_0+D_{y_1}\Phi_0h_1, \\
    (\forall t\ge 1)~z_t&=D_{y_{t-1}}\Phi_th_{t-1}+D_{y_t}\Phi_th_t+D_{y_{t+1}}\Phi_th_{t+1},
\end{align*}
where all $\Phi_t$'s are evaluated at $(x^*,y^*)$. Writing down the entries yields
\begin{subequations}\label{eq:zt}
\begin{align}
    z_{1,t}&=-g_k(k^*,p^*)h_{1,t-1}+h_{1,t}-g_p(k^*,p^*)h_{2,t}, \label{eq:zt1}\\
    z_{2,t}&=-\frac{f''(k^*)}{G}p^*h_{1,t}-\frac{f'(k^*)}{G}h_{2,t}+h_{2,t+1} \label{eq:zt2}
\end{align}
\end{subequations}
for $t\ge 0$ with the initial condition $h_{1,-1}=0$. Letting $w_t\coloneqq (h_{1,t-1},h_{2,t})$, we may rewrite \eqref{eq:zt} as
\begin{align}
z_t&=Lw_{t+1}-Mw_t, &
    L&\coloneqq \begin{bmatrix}
        1 & 0\\
        -f''p^*/G & 1
    \end{bmatrix}, & M&\coloneqq \begin{bmatrix}
        g_k & g_p\\
        0 & 1
    \end{bmatrix},
    \label{eq:wt1}
\end{align}
where all functions are evaluated at $(k^*,p^*)$ and we have used $f'(k^*)=G$. Since $L$ is invertible, we may rewrite \eqref{eq:wt1} as
\begin{equation}
    w_{t+1}=L^{-1}Mw_t+L^{-1}z_t\eqqcolon Aw_t+u_t.\label{eq:wt2}
\end{equation}

Let us verify that the system \eqref{eq:wt2} with the initial condition $w_{1,0}=0$ satisfies the assumptions of Lemma \ref{lem:lindiff}. We check the assumptions one by one.
\begin{itemize}
\item Using \eqref{eq:wt1}, the matrix $A$ in \eqref{eq:wt2} simplifies to
\begin{equation*}
    A\coloneqq L^{-1}M=
    \begin{bmatrix}
        1 & 0\\
        f''p^*/G & 1
    \end{bmatrix}\begin{bmatrix}
        g_k & g_p\\
        0 & 1
        
    \end{bmatrix}=\begin{bmatrix}
        g_k & g_p\\
        f''p^*g_k/G & f''p^*g_p/G+1
    \end{bmatrix}.
\end{equation*}
The characteristic function of $A$ is
\begin{equation*}
    q(\lambda)\coloneqq \lambda^2-(g_k+f''p^*g_p/G+1)\lambda+g_k.
\end{equation*}
By Lemma \ref{lem:system}, we have $q(0)=g_k>0$ and $q(1)=-f''p^*g_p/G<0$. Therefore, $A$ has two real eigenvalues $\lambda_1,\lambda_2$ satisfying $0<\lambda_1<1<\lambda_2$.
\item Since $\set{z_t}$ is a bounded sequence in $\R^2$ and $u_t=L^{-1}z_t$, so is $\set{u_t}$.
\item The initial value $w_0$ satisfies $w_{1,0}=h_{1,-1}=0$, which corresponds to setting $b=(1,0)$ and $c=0$ in Lemma \ref{lem:lindiff}.
\item We show $(bP)_1\neq 0$, where $P=(p_{ij})$ is the matrix that diagonalizes $A$ as in \eqref{eq:diagonal} and $p_{ij}$ is its $(i,j)$ entry. Suppose  $(bP)_1=0$. Since
\begin{equation*}
    bP=\begin{bmatrix}
        1 & 0
    \end{bmatrix}
    \begin{bmatrix}
        p_{11} & p_{12} \\
        p_{21} & p_{22}
    \end{bmatrix}
    =\begin{bmatrix}
        p_{11} & p_{12}
    \end{bmatrix},
\end{equation*}
we obtain $p_{11}=0$. Since $P$ is invertible, we have $p_{21}\neq 0$. By rescaling $P$ if necessary, we may assume $p_{21}=1$. Multiplying $P$ from the left to \eqref{eq:diagonal} and comparing the first column, we obtain
\begin{equation*}
    \lambda_1\begin{bmatrix}
        0 \\ 1
    \end{bmatrix}=A\begin{bmatrix}
        0 \\ 1
    \end{bmatrix}=\begin{bmatrix}
        g_p\\
        f''p^*g_p/G+1
    \end{bmatrix},
\end{equation*}
which contradicts $g_p<0$. Therefore, $(bP)_1\neq 0$.
\end{itemize}
By Lemma \ref{lem:lindiff}, there exists a unique bounded sequence $\set{w_t}$ in $\R^2$ satisfying \eqref{eq:wt2} with the initial condition $w_{1,0}=0$, so $D_y\Phi(x^*,y^*)$ is bijective.

Since $\Phi$ in \eqref{eq:Phi} is continuously Fr\'echet differentiable and $D_y\Phi(x^*,y^*)$ is invertible, by the implicit function theorem for Banach spaces (see Problem 2 in \citet[p.~266]{Luenberger1969} and \citet[Theorem 3.4.10]{KrantzParks2003}), there exist a constant $\delta>0$ and a continuous mapping $\phi:B_\delta(x^*)\to Y$ (where $B_\delta(x^*)\coloneqq \set{x\in X:\norm{x-x^*}\le \delta}$ is the $\delta$-ball) with $\phi(x^*)=y^*$ such that, for all $x\in B_\delta(x^*)$, we have $\Phi(x,y)=0$ if $y=\phi(x)$. Since $x=(x_t)$, $x_0=k_0$, and $x_t=d_t$ for $t\ge 1$, if $\abs{k_0-k^*}\le \delta$ and $\sup_{t\ge 1} d_t\le \delta$, then there exists a bounded sequence $y=\set{(k_{t+1},p_t)}$ such that the equilibrium conditions \eqref{eq:system} hold. Continuity of $\phi$ implies that $\set{(k_t,p_t)}$ is close to $\set{(k^*,p^*)}$, so we have $k_t>0$ and $p_t>0$. Thus, $\set{(k_t,p_t)}$ is an equilibrium. Since $\set{(k_t,p_t)}$ is close to $\set{(k^*,p^*)}$, it cannot be of the form of Lemma \ref{lem:longrun}\ref{item:lr_bubbleless}. Therefore, it must be of the form of Lemma \ref{lem:longrun}\ref{item:lr_bubbly}.
\end{proof}

\subsection{Proof of Theorem \ref{thm:restore2}}
Take $\delta>0$ as in Theorem \ref{thm:restore}. Since $\limsup_{t\to\infty}d_t^{1/t}<1$, we have $d_t\to 0$. Hence there exists $T\in \N$ such that $\sup_{t\ge T}d_t\le \delta$. By Theorem \ref{thm:restore}, for sufficiently large $k_T>0$, there exists a unique and asymptotically bubbly equilibrium $\set{(k_t,p_t)}_{t=T}^\infty$ starting at $T$ with initial capital $k_T$. For $t=T,\dots,1$, recursively define $(k_{t-1},p_{t-1})$ using Lemma \ref{lem:system}, or equivalently
\begin{subequations}\label{eq:recursive}
\begin{align}
    p_{t-1}&=\frac{G}{f'(k_t)}(p_t+d_t), \label{eq:recursive_p}\\
    s(\omega(k_{t-1}),f'(k_t))&=Gk_t+p_{t-1}, \label{eq:recursive_k}
\end{align}
\end{subequations}
where $s$ is the savings function and $\omega(k)\coloneqq f(k)-kf'(k)$. Clearly, such $\set{(k_t,p_t)}_{t=0}^{T-1}$ exist if the function $\psi(k)\coloneqq s(\omega(k),R)$ has range $(0,\infty)$ for any fixed $R>0$. Under this condition, by Lemmas \ref{lem:system} and \ref{lem:k0}, if we define $\kappa=k_0$, then the conclusion of Theorem \ref{thm:restore} holds for any $k_0\ge \kappa$, which proves Theorem \ref{thm:restore2}.

It thus remains to show that $\psi$ has range $(0,\infty)$. Since $\omega'(k)=-kf''(k)>0$ and Assumption \ref{asmp:U} holds, $\psi$ is continuous and strictly increasing. Hence it suffices to show $\psi(0)=0$ and $\psi(\infty)=\infty$. Using the trivial bound $0\le s(w,R)\le w$, we obtain $0\le \psi(k)\le \omega(k)$. Since $\omega$ has range $(0,\infty)$, we have $\omega(0)=0$ and hence $\psi(0)=0$. To show $\psi(\infty)=\infty$, since $\omega(\infty)=\infty$, it suffices to show $s(\infty,R)=\infty$. Taking the first-order condition of \eqref{eq:eq_utility}, we have
\begin{equation}
    R=\frac{U_1}{U_2}(w-s(w,R),Rs(w,R)). \label{eq:foc}
\end{equation}
If $s(\infty,R)\eqqcolon \bar{s}<\infty$, letting $w\to\infty$ in \eqref{eq:foc}, we obtain $R=(U_1/U_2)(\infty,R\bar{s})=0$ by \eqref{eq:sufficient}, which is a contradiction. Therefore, $s(\infty,R)=\infty$. \hfill \qedsymbol

\printbibliography

\newpage

\begin{center}
{\LARGE\bf Online Appendix (Not for publication)}
\end{center}

\begin{refsection}
\section{Systematic literature search}

We conducted a systematic literature search to identify bibliographic items related to Proposition 1 of \citet{Tirole1985}.

\subsection{Data collection}

On May 14, 2025, Toda's research assistant Johar Cassa (PhD student at Emory University) used the software \emph{Publish or Perish}\footnote{\url{https://harzing.com/resources/publish-or-perish}} to create a list of bibliographic items citing \citet{Tirole1985}. This resulted in 1{,}943 items, which is very close to the Google Scholar citation counts on the same day (1{,}964). We used \emph{Publish or Perish} because it made it easier to retrieve information such as publication year, author names, titles, publishers, URLs, etc.

We focused on items written in English, resulting in 1{,}592 items. The justification is that, as Proposition 1 of \citet{Tirole1985} is technical, if there is something scientifically significant related to it, the item is likely written in English.

Among the remaining 1{,}592 items, we checked 1{,}435 (90.1\%). The reasons we were unable to check some items include the deletion of old working papers, publication in obscure outlets that we do not have access to (typically books and book chapters), among others.

For each of the remaining items, we skimmed the text and assigned the dummy variable \texttt{Proposition1}, which takes the value 1 if the item discusses anything remotely related to Proposition 1 of \citet{Tirole1985} with a dividend-paying asset (either statement (a), (b), or (c)), even if the item does not explicitly mention Proposition 1. Among the 1{,}435 items we checked, 47 (3.3\%) had $\texttt{Proposition1}=1$. Our spreadsheet is available on Toda's website.\footnote{\url{https://alexisakira.github.io/files/tirole_citations.xlsx}}

\subsection{Evaluation}

We carefully read each item with $\texttt{Proposition1}=1$ and evaluated how Proposition 1 is (explicitly or implicitly) discussed. Below are our findings, where we list items (in chronological and then alphabetical order).

\begin{itemize}
    \item \citet[Footnote 9]{DavidsonMartin1991} cite Proposition 1 of \citet{Tirole1985} without a specific discussion.
    \item \citet{Rhee1991} considers an extension of the \citet{Tirole1985} model where land (a durable non-reproducible asset) enters the production function. \citet[p.~794]{Rhee1991} states ``In proving the possibility of dynamic inefficiency, this example generalizes Tirole's (1985) analysis of deterministic bubbles on assets yielding constant rents. [\ldots] additional restrictions are needed to obtain the uniqueness of a non-steady-state equilibrium path.'' In this model, land rent (marginal productivity of land) is endogenous, but \citet[Assumption A]{Rhee1991} directly imposes a high-level assumption. Under this assumption, his Proposition 2 discusses the long-run behavior of equilibrium. Proposition 2 of \citet{Rhee1991} closely parallels Proposition 1 of \citet{Tirole1985}, but we note the following two important points. First, \citet[Proposition 2]{Rhee1991} has only parts (a), (b), which correspond to parts (a), (b) of Proposition 1 of \citet{Tirole1985}. Second, the proof simply states ``The proof is basically equivalent to the proof of Proposition 1 in Tirole (1985). Assumption A is needed for the proof of Lemma 1 in his Appendix.'' without providing any details. We thus conclude that \citet{Rhee1991} does not refer to Proposition 1(c) and does not dispute the analysis of \citet{Tirole1985}.
    \item \citet*[Footnote 12]{DavidsonMartinMatusz1994} cite Proposition 1 of \citet{Tirole1985} without a specific discussion.
    \item \citet[p.~351]{Burke1996} refers to \citet[\S7]{Wilson1981} and \citet[Proposition 1(c)]{Tirole1985} as examples where asset prices necessarily include bubbles.
    \item \citet[Footnote 4]{Femminis2002} states ``most of Tirole's analysis is carried out assuming that the aggregate quantity of rent is exogenously fixed in terms of output'', referring to the analysis with a dividend-paying asset. However, there is no discussion of Proposition 1.
    \item \citet[Footnote 22]{Lauri2004} recognizes that \citet{Tirole1985} considered dividend-paying assets, though there is no discussion of Proposition 1.
    \item \citet[p.~184]{Binswanger2005} states ``Tirole (1985) shows that the results derived for bubbles on intrinsically useless assets can be generalized to assets paying a dividend as long as dividends grow at a slower rate than the economy'', referring to the analysis with a dividend-paying asset. Furthermore, \citet[p.~192]{Binswanger2005} states ``Proposition 1 can be compared to the conditions for the existence of bubbles in deterministic economies derived in Tirole (1985, p.~1504)'', which refers to Proposition 1 of \citet{Tirole1985}.
    \item \citet[Table 3.1]{Siwasarit2006} summarizes Proposition 1 of \citet{Tirole1985} nearly verbatim, but there is no discussion beyond that.
    \item \citet[p.~70]{BosiSeegmuller2013} state ``Tirole (1985) proves the existence of a unique and monotonic growth path which converges to a bubbly steady state'', which refers to Proposition 1(b) (as they consider pure bubbles).
    \item Several papers by Bosi, Le Van, Pham, and coauthors extensively discuss \citet{Tirole1985}.
    \begin{itemize}
        \item \citet*[p.~215]{BosiLeVanPham2016Bookchapter} study an infinite-horizon model with a dividend-paying asset and state ``As long as dividends tend to zero, the land price remains higher than this fundamental value'', without mentioning Proposition 1. 
        \item \citet{BosiPham2016} refer to Proposition 1 of \citet{Tirole1985} but only for the case without dividends (pure bubble).
        \item \citet*{BosiHa-HuyLeVanPhamPham2018} consider \citet{Tirole1985}'s model with altruism and positive dividends and characterize the long-run behavior in Proposition 2 (which corresponds to \citet[Proposition 1]{Tirole1985}). However, their Assumption 6 is a high-level assumption that is often violated in some settings. After their Proposition 3, they state ``It should be noticed that Tirole (1985) does not consider the case where $\liminf_{t\to\infty}k_t$ may be zero. However, this case may be possible.'' Indeed, Example 1 of \citet{BosiHa-HuyLeVanPhamPham2018} provides an example where $k_t$ converges to zero. %However, \citet[Remark 7]{PhamTodaWP} verify that this is a counterexample to Proposition 1(a) of \citet{Tirole1985}, not Proposition 1(c).
        \item \citet{BosiLeVanPham2022} consider a model with infinitely-lived agents and a dividend-paying asset and construct an example of a continuum of bubbly equilibria in Proposition 7, which relates to Proposition 1(b) of \citet{Tirole1985}.
    \end{itemize}
    \item \citet*[Footnote 8]{AllenBarlevyGale2017} state ``Tirole (1985) showed that if dividends are positive and the limiting interest rate without a bubble is nonpositive, the equilibrium is unique and features a bubble'', which obviously refers to Proposition 1(c) of \citet{Tirole1985}. Furthermore, \citet[\S2]{AllenBarlevyGale2017} present an OLG model with risk-neutral agents and a dividend-paying asset in which the unique equilibrium is bubbly, which is essentially the same as the example in \citet[\S7]{Wilson1981}. (A significantly revised version of \citet{AllenBarlevyGale2017} was published as \citet{AllenBarlevyGale2025}.)
    \item \citet[Footnote 29]{BassettoCui2018} state ``Adapting Proposition 1 in Tirole (1985), one can then prove that there exists a unique equilibrium, even though the interest rate is asymptotically negative'', which likely refers to Proposition 1(c). In the 2013 working paper version, this footnote is labeled 23 and the authors thank Gadi Barlevy for pointing this out.
    \item \citet[Footnote 6]{MartinVentura2018} state ``we have known since the work of Tirole (1985) that there exist environments in which the unique rational market psychology must feature a bubble'', referring to \citet{AllenBarlevyGale2017}.
    \item \citet[p.~205]{Sorger2019} refers to Proposition 1 of \citet{Tirole1985}. However, as he deals with an intrinsically worthless asset, the reference is obviously to Proposition 1(a)(b) and not (c).
    \item \citet[Footnote 4]{Galichere2022} states ``Deterministic bubble assets with positive fundamentals need to satisfy either one of both following conditions to exist: i) the total rent grows at a slower rate than the economy [\ldots]'', referring to the analysis with a dividend-paying asset. However, there is no discussion of Proposition 1.
    \item \citet*[p.~11]{MichauOnoSchlegl2023} refer to Proposition 1 of \citet{Tirole1985} in connection to their Proposition 1, whose bullet points correspond to statements (a), (b), (c), but there is no discussion.
    \item \citet[Footnote 16]{Plantin2023} states ``It would be straightforward to add a ``tree'' to which bubbles are attached, as in Tirole (1985)'', referring to the analysis with a dividend-paying asset. However, there is no discussion of Proposition 1.
    \item Several papers (published and unpublished) by Hirano and Toda point out some issues with \citet[Proposition 1]{Tirole1985}.
    \begin{itemize}
        \item The review article of \citet[\S5.2]{HiranoToda2024JME} states ``Proposition 1(c) of Tirole (1985) recognizes the possibility that bubbles are necessary for equilibrium existence if the interest rate without bubbles is negative. Although he gives some explanations on p.~1506 in the sentence starting with ``The intuition behind this fact roughly runs as follows'', he did not necessarily provide a formal proof. In Tirole (1985), the proof of the nonexistence of fundamental equilibria appears at the bottom of p.~1522 and the top of p.~1523. The proof uses a convergence result discussed in Lemma 2. However, this convergence heavily relies on the monotonicity condition on the function $\psi$ defined in Equation (7) on p.~1502. This monotonicity/stability condition is a high-level assumption that need not be satisfied in a general setting.'' In hindsight, this monotonicity/stability condition is not an issue; see \citet[\S C.1]{PhamTodaWP}, especially p.~48.
        \item \citet[Footnote 4]{HiranoToda2025JPE} refers to \citet[\S5.2]{HiranoToda2024JME} with the same logic. Theorem 3 of \citet[\S V.A]{HiranoToda2025JPE} proves the necessity of bubbles in Tirole's model (which corresponds to his Proposition 1(c)) under Cobb-Douglas utility and other assumptions. One of their key assumptions is that $\liminf_{t\to\infty}K_t>0$ (capital is bounded away from zero), which is an assumption on an endogenous object and hence does not resolve the issue.
        \item \citet[Footnote \P]{HiranoToda2025PNAS} point out issues in \citet[Proposition 1]{Tirole1985} and \citet[Proposition 2]{Rhee1991} and refer to \citet{PhamTodaWP}.
        \item Several other unpublished manuscripts refer to either \citet[\S5.2]{HiranoToda2024JME} or \citet{PhamTodaWP}.
    \end{itemize}
\end{itemize}

In summary, among the 1{,}943 bibliographic items citing \citet{Tirole1985}, excluding duplicate items and papers by Hirano and Toda, only 4 (0.2\%) discuss or exploit the specific conclusion of Proposition 1(c), which are \citet[p.~351]{Burke1996}, \citet*[Footnote 8]{AllenBarlevyGale2017}, \citet[Footnote 29]{BassettoCui2018}, and \citet[Footnote 6]{MartinVentura2018}. The authors of the last two items acknowledge they learned Proposition 1(c) from Gadi Barlevy. None of these authors disputes the analysis of \citet{Tirole1985}, except \citet[Example 1]{BosiHa-HuyLeVanPhamPham2018}, who raise issues with Proposition 1(a).

\printbibliography
\end{refsection}

\end{document}